\documentclass[conference]{IEEEtran}
\pdfoutput=1
\usepackage{cite}
\usepackage{spconf}
\usepackage{amsmath,amssymb,amsfonts}
\usepackage{graphicx}
\usepackage{textcomp}
\usepackage{subfigure}
\usepackage{xcolor}
\usepackage{amsthm}
\usepackage[noend]{algpseudocode}
\usepackage{tabularx}
\usepackage{booktabs}
\usepackage{svg}
\theoremstyle{definition}
\newtheorem{definition}{Definition}

\newtheorem{theorem}{Theorem}

\begin{document}
\title{Social Welfare Maximization in cross-silo Federated Learning}
%

%
%
\name{Jianan Chen$^{\dag}$ \qquad Qin Hu$^{\dag\star}$ \qquad Honglu Jiang$^{\ddag}$}
\address{$^{\dag}$ Department of Computer and Information Science \\
Indiana University Purdue University Indianapolis, Indianapolis, IN, USA\\
      $^{\ddag}$ Department of Informatics and Engineering Systems\\
      The University of Texas Rio Grande Valley, Brownsville, TX, USA
}

\maketitle
\begin{abstract}
As one of the typical settings of Federated Learning (FL), cross-silo FL allows organizations to jointly train an optimal Machine Learning (ML) model. In this case, some organizations may try to obtain the global model without contributing their local training, lowering the social welfare. In this paper, we model the interactions among organizations in cross-silo FL as a public goods game for the first time and theoretically prove that there exists a social dilemma where the maximum social welfare is not achieved in Nash equilibrium. To overcome this social dilemma, we employ the Multi-player Multi-action Zero-Determinant (MMZD) strategy to maximize the social welfare. With the help of the MMZD, an individual organization can unilaterally control the social welfare without extra cost. Experimental results validate that the MMZD strategy is effective in maximizing the social welfare.
\end{abstract}
                                                
\section{Introduction}
\label{sec:Intro}
In Federated Learning (FL), clients cooperatively train a Machine Learning (ML) model with their decentralized datasets under the coordination of a central server\cite{mcmahan2017communication}. One of the typical settings of FL is cross-silo FL\cite{kairouz2019advances} where a neutral third-party agent acts as the central server and clients are a group of organizations, aiming to jointly train an optimal ML model for their respective use. In this case, these organizations are also the owners of the global model and can utilize the well-trained global model to further process tasks for their own interests. 
\footnotemark[1]
\footnotetext[1]{This work is partly supported by the US NSF under grant CNS-2105004.}
\footnotetext[2]{${\star}$ Corresponding author.}
An optimal global model with high performance requires the organizations in cross-silo FL to collaborate efficiently so as to bring considerable benefits to all participants, which can be regarded as the maximization of the social welfare. In fact, there are many studies on optimizing the social welfare in cross-silo FL to directly improving the model performance\cite{huang2021personalized,wang2021efficient,nandury2021cross,majeed2020cross}, increasing the convergence speed\cite{marfoq2020throughput}, reducing the communication cost\cite{zhang2020batchcrypt}, protecting privacy\cite{heikkila2020differentially,li2021practical,long2021federated} and security\cite{jiang2021flashe}, etc.

However, since every organization participating in cross-silo FL can obtain the final global model regardless of its contribution, the well-trained model becomes a public good, which is non-excludable and non-rivalrous for all organizations\cite{tang2021incentive}. This leads to the emergence of selfish behaviors that some organizations may only consider their own interests via inactively participating in local training while obtaining the final global model for free or at a lower cost. The spread of this behavior can result in huge loss of the social welfare, and then none of the organizations can get the optimal model, which compromises the long-term stability and sustainability of cross-silo FL.

The existing studies motivate organizations to fully contribute to cross-silo FL by designing an incentive mechanism\cite{tang2021incentive}. However, this requires extra negotiation costs since organizations need to reach a consensus on the mechanism in advance, and demands additional running costs where a distributed algorithm runs over all organizations, clearly adding more burden to organizations.

In this paper, we take a brand-new approach using the Multi-player Multi-action Zero-Determinant (MMZD) strategy\cite{he2016zero} to maximize the social welfare in cross-silo FL without causing additional costs for organizations. Another outstanding advantage of our method is that it can control the social welfare and be applied to any cross-silo FL scenario no matter what strategies or actions organizations perform. In summary, our contributions include:

\noindent 1) We model the interactions among organizations in cross-silo FL as a public goods game for the first time, focusing on the organization's strategy rather than designing an extra mechanism to solve the social welfare maximization problem.

\noindent 2) We reveal the existence of the social dilemma in cross-silo FL by mathematical proof for the first time, which demonstrates the adverse effect of selfish behaviors in cross-silo FL from the perspective of game theory. This can be used as a theoretical basis for exploring organizations' behaviors in cross-silo FL.

\noindent 3) We overcome the social dilemma by employing the MMZD strategy from view of an individual organization. Specifically, any organization can unilaterally maximize the social welfare, which ensures the social welfare in cross-silo FL at a certain level and maintains the stability of the system. This approach can also extend the application domains of the MMZD.

\noindent 4) Experiments prove the effectiveness of the MMZD strategy in maximizing the social welfare.

\section{System Model} \label{Modeling}
We consider a cross-silo FL scenario with a set of organizations, denoted as $\mathcal{N}=\{1,2,...,N\}$. All organizations rely on a central server to collaboratively conduct global model training for a specific task, where each of them has their own data for local training. The goal of organizations is to obtain an optimal global model, minimizing the loss with all datasets. The central server collects the results of local model updates from all organizations, aggregates to obtain the global model, and then distributes it to everyone for the next round of training. 

In each round of local training, every organization performs $K$ iterations of model training. We denote the number of global aggregations as $r$. For the current task, the strategy of organization $i\in\mathcal{N}$, denoted as $y_i\in\{0,1,...,r\}$, represents the number of global aggregation rounds it participates in the task. Then, $\mathbf{y}=(y_1,...,y_i,...,y_N)$ denotes the strategy vector of all organizations. Here we assume that all organizations in this cross-silo FL may participate in fewer global aggregations due to laziness or selfishness, but they do not carry out malicious attacks, such as model poisoning attack.

According to the cross-silo FL model, all organizations get the same model in return. Inspired by \cite{tang2021incentive}, we define the revenue of organization $i$ as
\begin{small}
\begin{equation}\label{eq:revenue}
\Phi_i(\mathbf{y})=m_i(\chi_0-\chi(\mathbf{y})),
\end{equation}
\end{small}where $m_i$ denotes the unit revenue of organization $i$ by using the returned final model, $\chi_0$ denotes the precision of the untrained model, and $\chi(\mathbf{y})$ denotes the precision of the trained global model with corresponding strategy vector $\mathbf{y}$. $\chi(\mathbf{y})$ can be modeled as $\chi(\mathbf{y})=\frac{\theta_0}{\theta_1+K \sum_{i\in\mathcal{N}} y_i}$, with positive coefficients $\theta_0$ and $\theta_1$ \cite{li2019convergence} being derived based on the loss function, neural network, and local datasets. In particular, we have $\chi_0=\frac{\theta_0}{\theta_1}$. The revenue of each organization is proportional to the difference between the expected loss after $r$ rounds aggregation (i.e.,$\chi(\mathbf{y})$) and the minimum expected loss (i.e.,$\chi_0$)\cite{tang2021incentive}. As the number of total participation rounds increases, the marginal decrease of the difference reduces.

We define the cost of organization $i$ as $\Psi_i(y_i)=C^{i}_p+C^{i}_m$. The cost is composed of the organization's computation cost $C^{i}_p$ and its communication cost $C^{i}_m$. On the one hand, the computation cost $C^{i}_p=\beta_i K y_i$, where $\beta_i$ is a positive parameter, denoting the computation cost of each iteration in organization $i$'s local training\footnote{As \cite{tran2019federated} shows, $\beta_i=\frac{\alpha_i}{2} f^2_i D_i S_i$, where $\frac{\alpha_i}{2}$ is the effective capacitance coefficient of organization $i$'s computing chipset, $f_i$ denotes the calculation processing capacity, $D_i$ denotes the number of data units, and $S_i$ denotes the number of CPU cycles required by organization $i$ to process one data unit.}. On the other hand, the communication cost $C^{i}_m$ is defined as a parameter since we assume that every organization uploads its updates in each global aggregation round. If 
it chooses not to participate in global aggregation, it will submit a zero vector as updates.

Then the utility of organization $i$ is defined as
the difference between its revenue and cost:
\begin{small}
\begin{equation}\label{eq:utility}
U^i(\mathbf{y})=\Phi_i(\mathbf{y})-\Psi_i(y_i).
\end{equation}
\end{small}According to previous statements, we model the interactions among organizations as a $\emph{cross-silo FL game}$. 
\begin{definition}
(Cross-silo FL game). In the cross-silo FL game, organizations participating in cross-silo FL act as players, where their actions and utilities are $y_i$ and $U^i(\mathbf{y})$, respectively.
\end{definition}
The cross-silo FL game can be iterative since these organizations in cross-silo FL usually cooperate for a long time to finish multiple FL tasks. Each game round in the cross-silo FL game corresponds to a certain FL task. Moreover, the social welfare in cross-silo FL game can be denoted as the total utility of all organizations $i$, namely $\Sigma^N_{i=1} U^i(\mathbf{y})$. In the cross-silo FL game, we find that the social dilemma occurs if $\Phi_i(\mathbf{y})-C^{i}_p<0$, which can be summarized as below.
\begin{theorem}\label{th:SD}
(Social dilemma). If $\Phi_i(\mathbf{y})-C^{i}_p<0$, there exists a social dilemma in the cross-silo FL game.
\end{theorem}
\begin{proof}
The social dilemma forms when the Nash equilibrium is not the point of maximum social welfare. First, we study the Nash equilibrium of the cross-silo FL game. Referring to (\ref{eq:utility}), we can derive the derivative of $U^i$ as $m_i\frac{K \theta_0}{(\theta_1+K \Sigma y_i)^2}- \beta_i K$. Given $\Phi_i(\mathbf{y})-C^{i}_p<0$, we have $m_i\frac{K y_i\theta_0}{(\theta_1+K y_i)\theta_1}-\beta_i K y_i<0$, which leads to $m_i\frac{K \theta_0}{(\theta_1+K \Sigma y_i )^2}<m_i\frac{K\theta_0}{(\theta_1+K y_i)\theta_1}<\beta_i K$.
Thus, the derivative of $U^i$ is negative, and the utility function decreases monotonically with $y_i$. The Nash equilibrium strategy of each organization is $y_i=0$, so the Nash equilibrium point is $\mathbf{y}^{NE}=(0,0,\dots,0)$. Noted that, there is a natural and necessary premise of FL that the utility of the well-trained model must be positive. Thus, we prove that the point $\mathbf{y}^{r}=(y_i=r,i\in\mathcal{N})$ with the social welfare $\Sigma^N_{i=1}U^i(\mathbf{y}^{r})=\sum \Psi_i(\mathbf{y}^{r})-\sum C^i_p-\sum C^i_m>0$ higher than that in the Nash equilibrium point $\Sigma^N_{i=1}U^i(\mathbf{y}^{NE})=-\sum C^i_m<0$. So the social dilemma exists if $\Phi_i(\mathbf{y})-C^{i}_p<0$.
\end{proof}

The condition $\Phi_i(\mathbf{y})-C^{i}_p<0$ in the above theorem indicates that if any organization $i\in\mathcal{N}$ only trains the local model using its local dataset, the utility is negative. In fact, this condition strengthens organizations' motivation to participate in global training in cross-silo FL game. 

\section{Social welfare maximization}\label{maxwelfare}
According to the analysis above, we can see that the underlying cause of the social dilemma is selfishness, leading to the loss of all organizations, namely the low social welfare. Aiming to solve this problem, we resort to the Multi-player Multi-action Zero-Determinant (MMZD) strategy for the social welfare maximization in this section. In each game round, any organization can choose the action $y_i\in\{0,1,...,r\}$, so there are $(r+1)^N$ possible outcomes for each game round. We assume that the organizations have one-round-memory since a long-memory player has no priority against others with short memory\cite{he2016zero}. For arbitrary organization $i\in\mathcal{N}$, its mixed strategy $\mathbf{p}^i$ is defined as:
\begin{small}
\begin{equation}
\mathbf{p}^i=[p^i_{1,0},p^i_{1,1},...,p^i_{1,r},p^i_{2,0},...,p^i_{j,g},...,p^i_{(r+1)^{N},r}]^T,
\end{equation}
\end{small}where $p^i_{j,g} (j\in\{1,2,...,(r+1)^N\},g\in\{0,1,...,r\})$ represents the probability of organization $i$ choosing action $y_i=g$ in the current game round and other organizations' choosing the same actions as $j$-th outcome of the previous game round. In addition, the corresponding utility vector $\mathbf{u}^i$ is denoted as:
\begin{small}
\begin{equation}
\mathbf{u}^i=[u^i_{1,0},u^i_{1,1},...,u^i_{1,r},u^i_{2,0},...,u^i_{j,g},...,u^i_{(r+1)^{N},r}]^T,
\end{equation}
\end{small}where each utility $u^i_{j,g}$ of organization $i$ choosing action $y_i=g$ in the $j$-th outcome can be calculated by $u^i_{j,g}=U^i(\mathbf{y}^{(j,g)})$, with $\mathbf{y}^{(j,g)}$ denoting the action vector $\mathbf{y}$ corresponding to the $j$-th outcome but $y_i=g$.

In the cross-silo FL model, an organization’s current move depends only on its last action and the strategy vector $\mathbf{y}$ in the last game round. We can construct a Markov matrix $\mathbf{M}=[M_{vw}]_{(r+1)^{N} \times (r+1)^{N}}$, with each element $M_{vw}$ denoting the one-step transition probability from state $v$ to $w$.

As demonstrated in\cite{he2016zero}, a certain column $\hat{\mathbf{p}}^i$ of $\mathbf{M}'\equiv \mathbf{M}-\mathbf{I}$ is only controlled by organization $i$'s strategy. We assume the stationary vector of $\mathbf{M}$ is $\mathbf{v}$, then organization $i$'s expected utility in the stationary state is:
\begin{small}
\begin{equation*}\label{eq:u_a}
E^i=\frac{\mathbf{v}\cdot\mathbf{u}^i}{\mathbf{v}\cdot\mathbf{1}}=\frac{\det (\mathbf{p}^1,\dots,\mathbf{p}^N,\mathbf{u}^i)}{\det (\mathbf{p}^1,\dots,\mathbf{p}^N,\mathbf{1})},
\end{equation*}
\end{small}which makes a linear combination of all organizations’ expected utilities yielding the following equation:
\begin{small}
\begin{equation}\label{eq:zd_equation}
\sum_{x=1}^{N}\alpha_x E^x+\alpha_0=\frac{\det (\mathbf{p}^1,\dots,\mathbf{p}^N,\sum_{x=1}^{N}\alpha_x \mathbf{u}^x+\alpha_0\mathbf{1})}{\det (\mathbf{p}^1,\dots,\mathbf{p}^N,\mathbf{1})}.
\end{equation}
\end{small}

In the above equation, $\alpha_0$ and $\alpha_x(x\in\mathcal{N})$ are constants. Thus, when organization $i$ chooses a strategy satisfying $\hat{\mathbf{p}}^i=\phi(\sum_{x=1}^{N}\alpha_x \mathbf{u}^x+\alpha_0\mathbf{1})$, where $\phi$ is a non-zero constant and $\hat{\mathbf{p}}^i$ is under the control of organization $i$, the corresponding column of $\hat{\mathbf{p}}^i$ and the last column of $\det (\mathbf{p}^1,\dots,\mathbf{p}^N,\sum_{x=1}^{N}\alpha_x \mathbf{u}^x+\alpha_0\mathbf{1})$ will be proportional. (\ref{eq:zd_equation}) becomes:
\begin{small}
\begin{equation}\label{eq:zd_zero}
\sum_{x=1}^{N}\alpha_x E^x+\alpha_0=0.
\end{equation}
\end{small}

We further study the social welfare maximization problem with the help of the MMZD in this circumstance. Take organization $1$ performing the MMZD strategy as an example. According to (\ref{eq:zd_zero}), by setting $\alpha_x=1(x\in\mathcal{N})$, the social welfare can be calculated as $\sum_{x=1}^{N} E^x=-\alpha_0$. Thus, the issue of maximizing the social welfare is equivalent to the following optimization problem:
\begin{small}
\begin{equation*}\label{eq:min_gamma}
\begin{split}
&\min \alpha_0,\\
&\,s.t.
\begin{cases}
0\le p^1_{j,g}\le 1,j\in\{1,2,...,(r+1)^N\},g\in\{0,...,r\},\\
\hat{\mathbf{p}}^1=\phi(\sum_{x=1}^{N} \mathbf{u}^x+\alpha_0\mathbf{1}),\\
\phi \neq 0.\\
\end{cases}
\end{split}
\end{equation*}
\end{small}We denote $u^x_k,k\in\{1,2,...,(r+1)^{N+1}\}$ as the $k$th element in $\mathbf{u}^x$, then we can solve the above optimization problem by considering the following two cases:\\
1) $\phi>0$. To meet the constraint $p^1_{j,g}\ge0$, we can get the lower bound of $\alpha_0$ as follows:
\begin{small}
\begin{align*}
& {\alpha_0}_{min}=\max(\Lambda_k),\forall{k}\in\{1,2,...,(r+1)^{N+1}\},\\
& \Lambda_k=
\begin{cases}   
-\sum_{x=1}^{N} u^x_k-\frac{1}{\phi}, &k=1,2,...(r+1)^{N-1},\\
-\sum_{x=1}^{N} u^x_k, & k=(r+1)^{N-1}+1,...,(r+1)^{N+1}.\\
\end{cases}
\end{align*}
\end{small}To meet the constraint $p^1_{j,g}\le1$, we can get the upper bound of $\alpha_0$ as follows:
\begin{small}
\begin{align*}
& {\alpha_0}_{max}=\min(\Lambda_l),\forall{l}\in\{(r+1)^{N+1}+1,...,2(r+1)^{N+1}\},\\
& \Lambda_l= \Lambda_{k+(r+1)^N}\\
& =
\begin{cases}
-\sum_{x=1}^{N} u^x_k, &k=1,2,...(r+1)^{N-1},\\
-\sum_{x=1}^{N} u^x_k+\frac{1}{\phi}, & k=(r+1)^{N-1}+1,...,(r+1)^{N+1}.\\
\end{cases}
\end{align*}
\end{small}Only if ${\alpha_0}_{min}\le{\alpha_0}_{max}$, can $\alpha_0$ have a feasible solution, which is equivalent to $\max(\Lambda_k)\le \min(\Lambda_l),\forall{k}\in\{1,2,...,(r+1)^{N+1}\},\forall{l}\in\{(r+1)^N+1,...,2(r+1)^{N+1}\}$. If there exists $\phi>0$ satisfying the above constraint, we can obtain the minimum value of $\alpha_0$ as follow:
\begin{small}
\begin{multline}\label{eq:case1_gammamin}
{\alpha_0}_{min}=\max\{-\sum_{x=1}^{N} u^x_1-\frac{1}{\phi},...,-\sum_{x=1}^{N} u^x_{(r+1)^{N-1}}-\frac{1}{\phi},\\
-\sum_{x=1}^{N} u^x_{(r+1)^{N-1}+1},...,-\sum_{x=1}^{N} u^x_{(r+1)^{N+1}}\}.
\end{multline}
\end{small}2) $\phi<0$. Similarly, when $p^1_{j,g}\ge0$, we have ${\alpha_0}_{min}=\max(\Lambda_l),\forall{l}\in\{(r+1)^{N+1}+1,...,2(r+1)^{N+1}\}$; considering $p^1_{j,g}\le1$, we have ${\alpha_0}_{max}=\min(\Lambda_k),\forall{k}\in\{1,2,...,(r+1)^{N+1}\}$. In addition, $\alpha_0$ is feasible only when ${\alpha_0}_{min}\le{\alpha_0}_{max}$, i.e., $\max(\Lambda_l)\le \min(\Lambda_k),\forall{k}\in\{1,2,...,(r+1)^{N+1}\},\forall{l}\in\{(r+1)^{N+1}+1,...,2(r+1)^{N+1}\}$. Finally, we can get the following result:
\begin{small}
\begin{multline}\label{eq:case2_gammamin}
{\alpha_0}_{min}=\max\{-\sum_{x=1}^{N} u^x_1,...,-\sum_{x=1}^{N} u^x_{(r+1)^{N-1}},\\
-\sum_{x=1}^{N} u^x_{(r+1)^{N-1}+1}+\frac{1}{\phi},...,-\sum_{x=1}^{N} u^x_{(r+1)^{N+1}}+\frac{1}{\phi}\}.
\end{multline}
\end{small}In summary, by (\ref{eq:case1_gammamin}) and (\ref{eq:case2_gammamin}), organization $1$ can unilaterally set the expected social welfare $\sum_{x=1}^{N} E^x$ with the MMZD strategy $\mathbf{p}^1$ meeting $\hat{\mathbf{p}}^1=\phi(\sum_{x=1}^{N} \mathbf{u}^x+\alpha_0\mathbf{1})$, with each element of $\mathbf{p}^1$ calculated by:
\begin{small}
\begin{align*}\label{eq:p_i_calculate}
p^1_h=
\begin{cases}   
\sum_{x=1}^{N} u^x_h+{\alpha_0}_{min}+1, h=1,2,...,(r+1)^{N-1},\\
\sum_{x=1}^{N} u^x_h+{\alpha_0}_{min}, h=(r+1)^{N-1}+1,...,(r+1)^{N+1},\\
\end{cases}
\end{align*}
\end{small}where $p^1_h$ denotes the $h$-th element in $\mathbf{p}^1$.
\section{Experiments}
In this section, we evaluate the performance of the MMZD strategy to maximize the social welfare based on simulation experiments. All experiments are implemented using Matlab R2021a on a laptop with 2.3 GHz Intel Core i5-8300H processor. We set $N=10$ since the number of organizations in cross-silo FL is usually small. We consider that $K=200$, $r=33$. $\theta_0=23271.584$ and $\theta_1=50193.243$ are derived from the simulation dataset. In addition, we set $\phi=0.01$. For every control group with different strategy settings, we repeat the experiment for 100 times. 
\begin{figure}[h]
\centering
  \centering
  \subfigure[MMZD]{
    \label{fig:subfig:a}
    \includegraphics[scale=0.275]{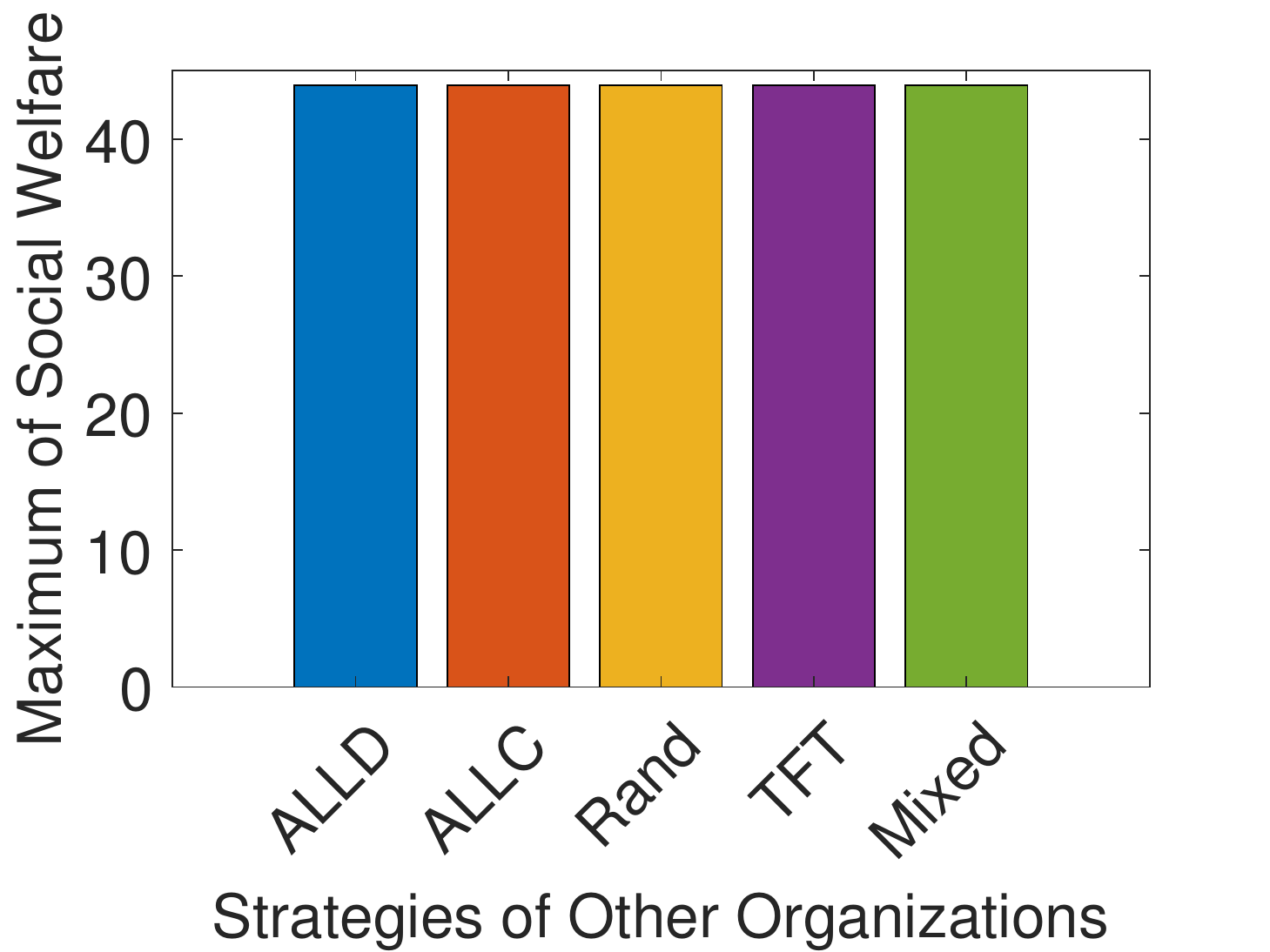}}
  \subfigure[ALLD]{
    \label{fig:subfig:b}
    \includegraphics[scale=0.275]{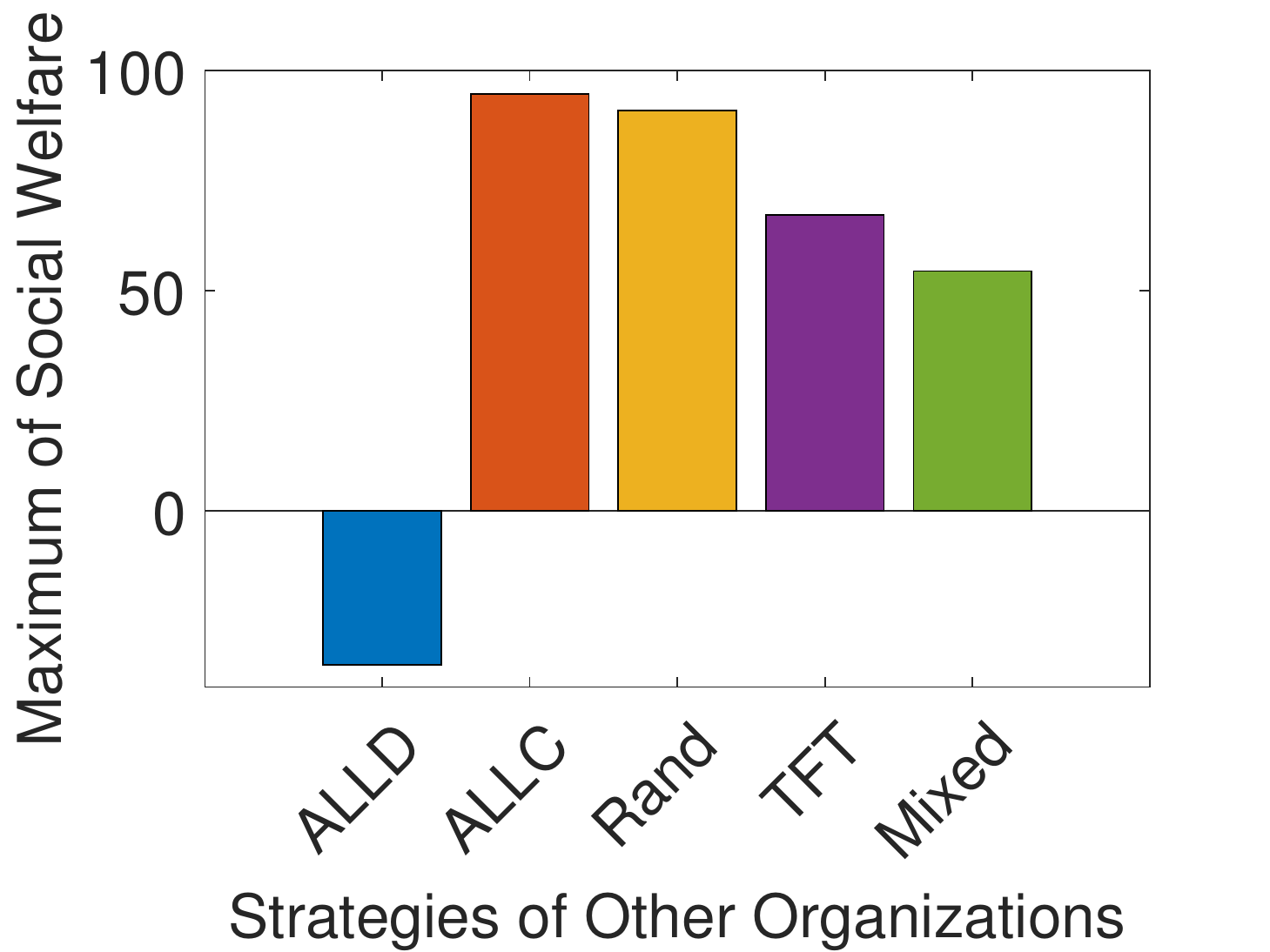}}
  \subfigure[ALLC]{
    \label{fig:subfig:c}
    \includegraphics[scale=0.275]{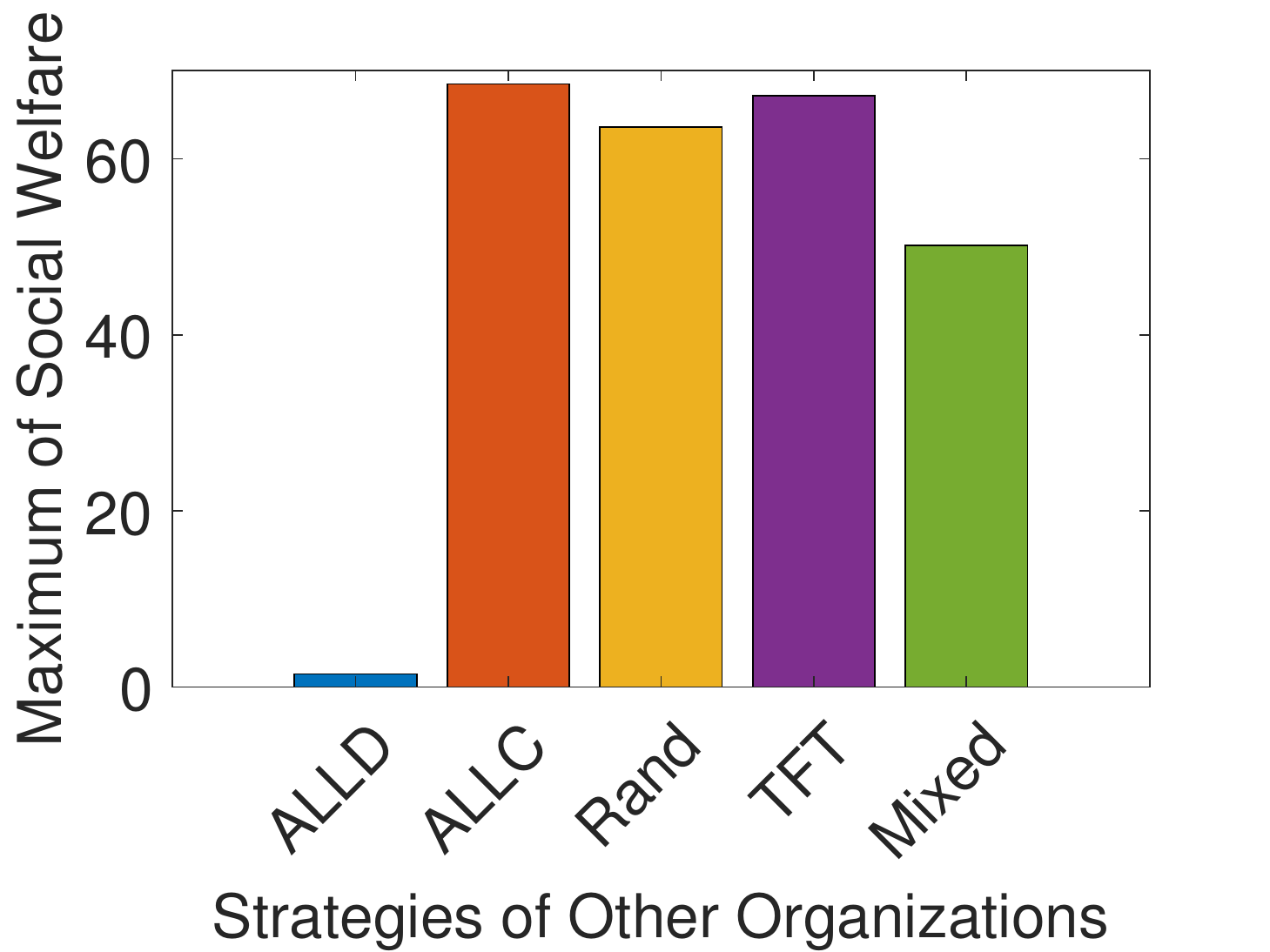}}
  \subfigure[Rand]{
    \label{fig:subfig:d}
    \includegraphics[scale=0.275]{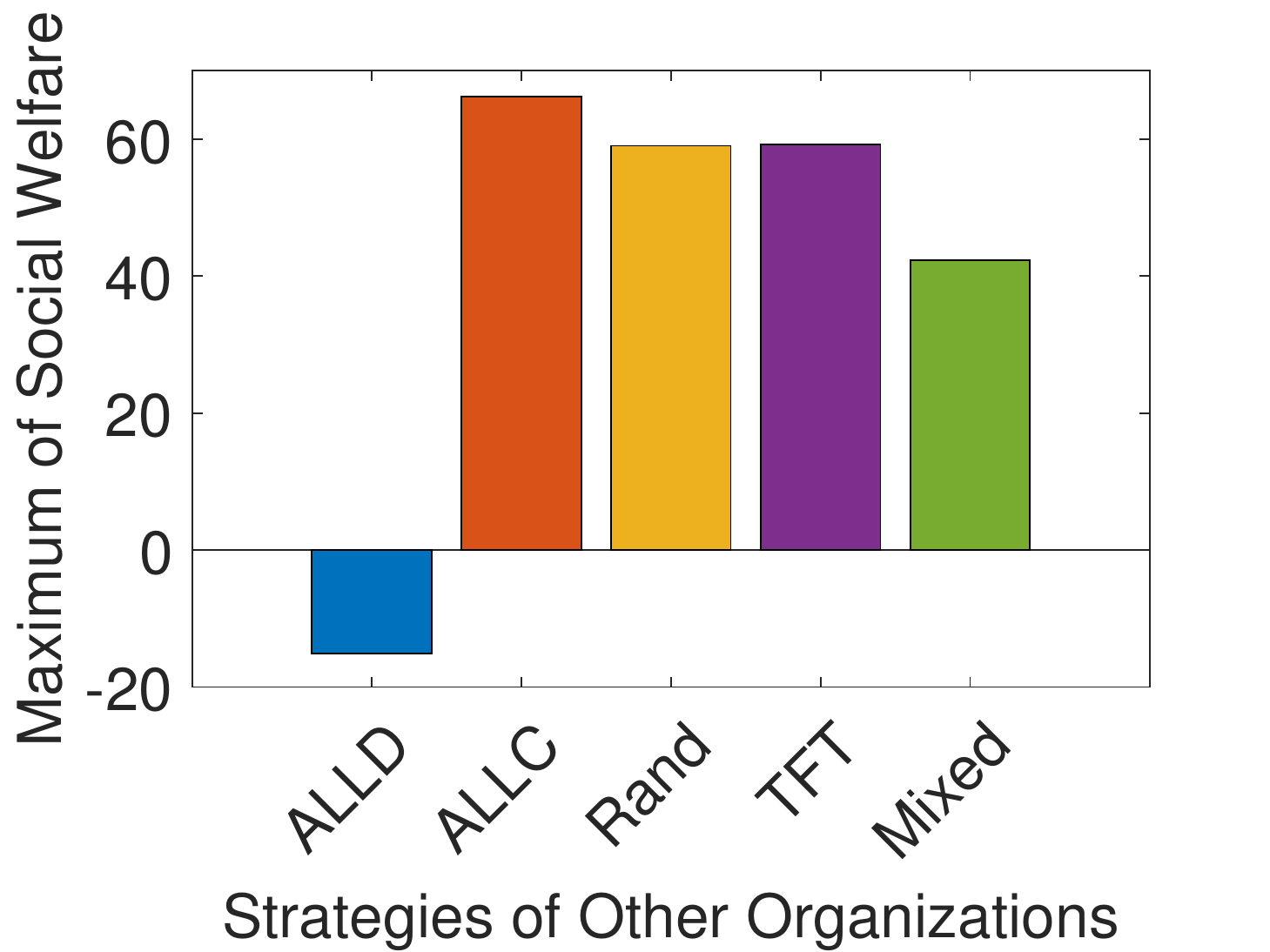}}
\caption{The maximum social welfare under different strategy combinations of organization 1 and others.}
  \label{fig:strategies} 
\end{figure}
In order to verify the effectiveness of the MMZD strategy on maximizing the social welfare, we compare it with other five classical strategies, by simulating the entire cross-silo FL process for 20 rounds of game. The numerical results of maximum social welfare denote the sum of all organizations’ utilities. In Fig. \ref{fig:strategies}, organization 1 adopts MMZD, all-defection (ALLD)\cite{hu2019quality}, all-cooperation (ALLC)\cite{hu2019quality}, and random (Rand) strategies\cite{hu2019quality}. Other organizations use ALLD, ALLC, Rand, Tit-For-Tat (TFT)\cite{nowak1993strategy}, and mixed (Mixed) strategies\cite{hu2019quality}. Specifically, ALLD strategy is defined as: the organization does not perform local training at all, only submits a zero vector in global aggregation. While ALLC strategy means that the organization participates in all $r$ global aggregation with their local updates in every game round. Organizations which adopt Rand strategy randomly participate in global aggregation from $0$ to $r$ rounds with the probability of $\frac{1}{r+1}$. TFT strategy is defined as the organizations randomly choose the number of participating global aggregation from $0$ to $\frac{\lfloor r \rfloor}{2}$ when the sum of global aggregation rounds in last game round is less than $\frac{Nr}{2}$, otherwise they randomly choose the number of participating global aggregation from $\frac{\lfloor r+1 \rfloor}{2}$ to $r$. We define the mixed strategy as adopting a specific strategy chosen from ALLC, ALLD, Rand, and TFT. By comparing Fig. \ref{fig:strategies}(a) with the other three figures, we can find that the MMZD strategy can effectively control the social welfare. This can prove that free-riding behavior is avoided in some ways.
\begin{figure}
\centering
\includegraphics[scale=0.5]{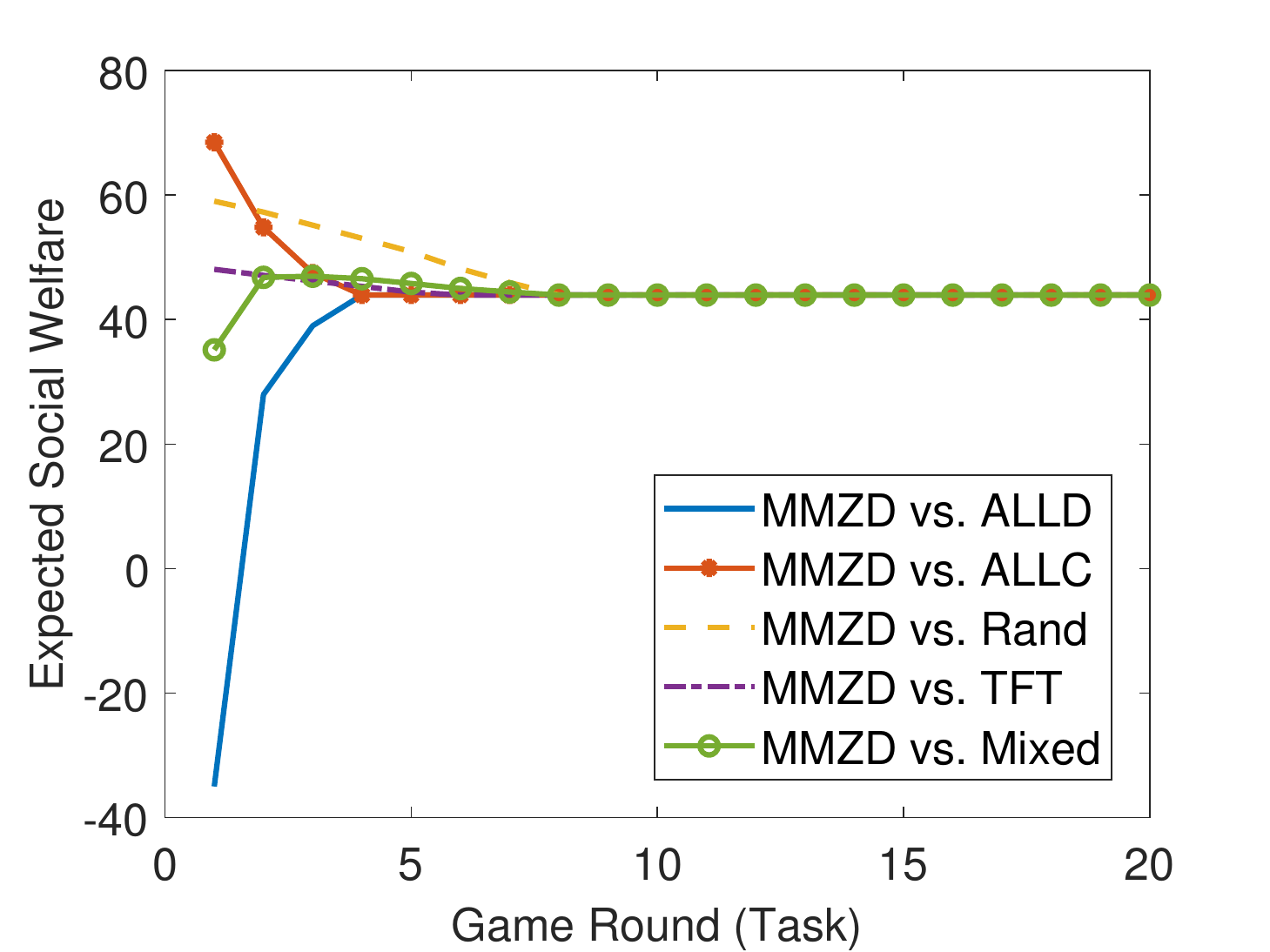}
\caption{Evolution of the expected social welfare.}
\label{fig:ZDconverge} 
\end{figure}

Fig. \ref{fig:ZDconverge} plots the expected social welfare changes in each game round, as the organization 1 adopts the MMZD strategy and other organizations employ different strategies. Fig. \ref{fig:strategies}(a) displays the final result in Fig. \ref{fig:ZDconverge}, which indicates that no matter what kind of strategies other organizations adopt, the social welfare finally converges to a fixed value, verifying the power of the proposed social welfare maximization game.

%
%
%

\section{conclusion}
In this paper, we define the cross-silo FL game among organizations as a public goods game, revealing the social dilemma in cross-silo FL theoretically. In order to overcome the social dilemma, we propose a brand-new method using the MMZD to solve the social welfare maximization problem. By the means of the MMZD, an individual organization can unilaterally control social welfare at a certain level, regardless of other organizations' strategies. Moreover, our approach can maintain the stability and sustainability of the system without extra cost. Simulation results prove that the MMZD strategy can efficiently and effectively control social welfare, which reduces the loss from selfish behaviors.
\label{sec:prior}

\vfill\pagebreak
\label{sec:refs}
\bibliographystyle{IEEEbib}
\bibliography{./refs.bib,./IEEEexample}

\begin{thebibliography}{10}

\bibitem{mcmahan2017communication}
Brendan McMahan, Eider Moore, Daniel Ramage, Seth Hampson, and Blaise~Aguera
  y~Arcas,
\newblock ``Communication-efficient learning of deep networks from
  decentralized data,''
\newblock in {\em Artificial intelligence and statistics}. PMLR, 2017, pp.
  1273--1282.

\bibitem{kairouz2019advances}
Peter Kairouz, H~Brendan McMahan, Brendan Avent, Aur{\'e}lien Bellet, Mehdi
  Bennis, Arjun~Nitin Bhagoji, Kallista Bonawitz, Zachary Charles, Graham
  Cormode, Rachel Cummings, et~al.,
\newblock ``Advances and open problems in federated learning,''
\newblock {\em arXiv preprint arXiv:1912.04977}, 2019.

\bibitem{huang2021personalized}
Yutao Huang, Lingyang Chu, Zirui Zhou, Lanjun Wang, Jiangchuan Liu, Jian Pei,
  and Yong Zhang,
\newblock ``Personalized cross-silo federated learning on non-iid data,''
\newblock in {\em Proceedings of the AAAI Conference on Artificial
  Intelligence}, 2021, vol.~35, pp. 7865--7873.

\bibitem{wang2021efficient}
Yansheng Wang, Yongxin Tong, Dingyuan Shi, and Ke~Xu,
\newblock ``An efficient approach for cross-silo federated learning to rank,''
\newblock in {\em 2021 IEEE 37th International Conference on Data Engineering
  (ICDE)}. IEEE, 2021, pp. 1128--1139.

\bibitem{nandury2021cross}
Kishore Nandury, Anand Mohan, and Frederick Weber,
\newblock ``Cross-silo federated training in the cloud with diversity scaling
  and semi-supervised learning,''
\newblock in {\em ICASSP 2021-2021 IEEE International Conference on Acoustics,
  Speech and Signal Processing (ICASSP)}. IEEE, 2021, pp. 3085--3089.

\bibitem{majeed2020cross}
Umer Majeed, Latif~U Khan, and Choong~Seon Hong,
\newblock ``Cross-silo horizontal federated learning for flow-based
  time-related-features oriented traffic classification,''
\newblock in {\em 2020 21st Asia-Pacific Network Operations and Management
  Symposium (APNOMS)}. IEEE, 2020, pp. 389--392.

\bibitem{marfoq2020throughput}
Othmane Marfoq, Chuan Xu, Giovanni Neglia, and Richard Vidal,
\newblock ``Throughput-optimal topology design for cross-silo federated
  learning,''
\newblock {\em arXiv preprint arXiv:2010.12229}, 2020.

\bibitem{zhang2020batchcrypt}
Chengliang Zhang, Suyi Li, Junzhe Xia, Wei Wang, Feng Yan, and Yang Liu,
\newblock ``Batchcrypt: Efficient homomorphic encryption for cross-silo
  federated learning,''
\newblock in {\em 2020 $\{$USENIX$\}$ Annual Technical Conference
  ($\{$USENIX$\}$$\{$ATC$\}$ 20)}, 2020, pp. 493--506.

\bibitem{heikkila2020differentially}
Mikko~A Heikkil{\"a}, Antti Koskela, Kana Shimizu, Samuel Kaski, and Antti
  Honkela,
\newblock ``Differentially private cross-silo federated learning,''
\newblock {\em arXiv preprint arXiv:2007.05553}, 2020.

\bibitem{li2021practical}
Qinbin Li, Bingsheng He, and Dawn Song,
\newblock ``Practical one-shot federated learning for cross-silo setting,''
\newblock in {\em Proceedings of the Thirtieth International Joint Conference
  on Artificial Intelligence (IJCAI-21)}, 2021.

\bibitem{long2021federated}
Guodong Long, Tao Shen, Yue Tan, Leah Gerrard, Allison Clarke, and Jing Jiang,
\newblock ``Federated learning for privacy-preserving open innovation future on
  digital health,''
\newblock {\em arXiv preprint arXiv:2108.10761}, 2021.

\bibitem{jiang2021flashe}
Zhifeng Jiang, Wei Wang, and Yang Liu,
\newblock ``Flashe: Additively symmetric homomorphic encryption for cross-silo
  federated learning,''
\newblock {\em arXiv preprint arXiv:2109.00675}, 2021.

\bibitem{tang2021incentive}
Ming Tang and Vincent~WS Wong,
\newblock ``An incentive mechanism for cross-silo federated learning: A public
  goods perspective,''
\newblock in {\em IEEE INFOCOM 2021-IEEE Conference on Computer
  Communications}. IEEE, 2021, pp. 1--10.

\bibitem{he2016zero}
Xiaofan He, Huaiyu Dai, Peng Ning, and Rudra Dutta,
\newblock ``Zero-determinant strategies for multi-player multi-action iterated
  games,''
\newblock {\em IEEE Signal Processing Letters}, vol. 23, no. 3, pp. 311--315,
  2016.

\bibitem{li2019convergence}
Xiang Li, Kaixuan Huang, Wenhao Yang, Shusen Wang, and Zhihua Zhang,
\newblock ``On the convergence of fedavg on non-iid data,''
\newblock {\em arXiv preprint arXiv:1907.02189}, 2019.

\bibitem{tran2019federated}
Nguyen~H Tran, Wei Bao, Albert Zomaya, Minh~NH Nguyen, and Choong~Seon Hong,
\newblock ``Federated learning over wireless networks: Optimization model
  design and analysis,''
\newblock in {\em IEEE INFOCOM 2019-IEEE Conference on Computer
  Communications}. IEEE, 2019, pp. 1387--1395.

\bibitem{hu2019quality}
Qin Hu, Shengling Wang, Peizi Ma, Xiuzhen Cheng, Weifeng Lv, and Rongfang Bie,
\newblock ``Quality control in crowdsourcing using sequential zero-determinant
  strategies,''
\newblock {\em IEEE Transactions on Knowledge and Data Engineering}, vol. 32,
  no. 5, pp. 998--1009, 2019.

\bibitem{nowak1993strategy}
Martin Nowak and Karl Sigmund,
\newblock ``A strategy of win-stay, lose-shift that outperforms tit-for-tat in
  the prisoner's dilemma game,''
\newblock {\em Nature}, vol. 364, no. 6432, pp. 56--58, 1993.

\end{thebibliography}

\end{document}